\documentclass[final,onefignum,onetabnum]{siamart171218}

\usepackage{amssymb}

\usepackage{braket,amsfonts}

\usepackage{array}

\usepackage[caption=false]{subfig}
\usepackage{pgfplots}

\newsiamthm{claim}{Claim}
\newsiamremark{remark}{Remark}
\newsiamremark{hypothesis}{Hypothesis}
\crefname{hypothesis}{Hypothesis}{Hypotheses}

\usepackage{algorithmic}

\usepackage{graphicx,epstopdf}

\Crefname{ALC@unique}{Line}{Lines}

\usepackage{amsopn}

\usepackage{xspace}
\usepackage{bold-extra}
\usepackage[most]{tcolorbox}

\colorlet{texcscolor}{blue!50!black}
\colorlet{texemcolor}{red!70!black}
\colorlet{texpreamble}{red!70!black}
\colorlet{codebackground}{black!25!white!25}

\lstdefinestyle{siamlatex}{%
	style=tcblatex,
	texcsstyle=*\color{texcscolor},
	texcsstyle=[2]\color{texemcolor},
	keywordstyle=[2]\color{texemcolor},
	moretexcs={cref,Cref,maketitle,mathcal,text,headers,email,url},
}

\tcbset{%
	colframe=black!75!white!75,
	coltitle=white,
	colback=codebackground, 
	colbacklower=white, 
	fonttitle=\bfseries,
	arc=0pt,outer arc=0pt,
	top=1pt,bottom=1pt,left=1mm,right=1mm,middle=1mm,boxsep=1mm,
	leftrule=0.3mm,rightrule=0.3mm,toprule=0.3mm,bottomrule=0.3mm,
	listing options={style=siamlatex}
}

\newtcblisting[use counter=example]{example}[2][]{%
	title={Example~\thetcbcounter: #2},#1}

\newtcbinputlisting[use counter=example]{\examplefile}[3][]{%
	title={Example~\thetcbcounter: #2},listing file={#3},#1}

\DeclareTotalTCBox{\code}{ v O{} }
{ 
	fontupper=\ttfamily\color{black},
	nobeforeafter,
	tcbox raise base,
	colback=codebackground,colframe=white,
	top=0pt,bottom=0pt,left=0mm,right=0mm,
	leftrule=0pt,rightrule=0pt,toprule=0mm,bottomrule=0mm,
	boxsep=0.5mm,
	#2}{#1}

\begin{tcbverbatimwrite}{tmp_\jobname_header.tex}
	\title{Guide to Using SIAM's \LaTeX\ Style\thanks{Submitted to the editors DATE.
			\funding{Funding information goes here.}}}
	
	\author{Dianne Doe\thanks{Imagination Corp., Chicago, IL (\email{ddoe@imag.com}).}
		\and Paul T. Frank\thanks{Department of Applied Math, Fictional University, Boise, ID (\email{ptfrank@fictional.edu}, \email{jesmith@fictional.edu}).}
		\and Jane E. Smith\footnotemark[3]}
	
	\headers{Guide to Using  SIAM'S \LaTeX\ Style}{Dianne Doe, Paul T. Frank, and Jane E. Smith}
\end{tcbverbatimwrite}
\input{tmp_\jobname_header.tex}

\ifpdf
\hypersetup{ pdftitle={Guide to Using  SIAM'S \LaTeX\ Style} }
\fi

\patchcmd\newpage{\vfil}{}{}{}
\begin{document}

		

		\title{\bf A practical Single Source Shortest Path algorithm for random directed graphs with arbitrary weight in expecting linear time} 
		\date{}
		\author{Li Dexin\thanks{Fujian Zhangzhou No.1 High School (\email{1171000806@stu.hit.edu.cn}).}
			}
		
		\renewcommand{\thefootnote}{\fnsymbol{footnote}}
		
		\maketitle

		{\noindent\small\section*{Abstract}
			In this paper, I present an algorithm called Raffica algorithm for Single-Source Shortest Path(SSSP). On random graph, this algorithm has linear time complexity(in expect). More precisely, the random graph uses configuration model, and the weights are distributed mostly positively. It is also linear for random grid graphs. Despite I made an assumption on the weights of the random graph, this algorithm is able to solve SSSP with arbitrary weights; when a negative cycle exists, this algorithm can find it out once traversed. The algorithm has a lot of appliances. }
		
		\vspace{1ex}
		{\noindent\small{\bf Keywords:}
			Time complexity; Negative cycles; Single Source Shortest Path; Raffica algorithm; Random Graph. }
		\section{Introduction}
		Single Source Shortest Path problem(SSSP) is the most basic problem in graph optimization studied in computer science. It is widely used in mathematical modeling, such as traffic regulation, \textbf{Systems of Difference Constraints}, and so on.
		\par\textbf{Prior work} Dijkstra gave the nonnegetive-weight problem a $\Theta(N^2)$ solution\cite{Dijkstra1959A}. This algorithm sorts the distance of the vertices. Based on Dijkstra's algorithm, using priority queues can make a new approach $\Theta(M\lg N)$. Using Fibonacci Heap will make it only cost $\Theta(M + N\lg N)$, which is almost linear. We can also use a bucket to solve it in $\Theta(M+W)$ \cite{introductionalgorithm}, where $W$ is the max weight of the graph. We can even use the characteristic of RAM to solve it in $\Theta(M\lg N \lg N)$ or $\Theta(M + N (\lg N)^{1/3+\epsilon})$\cite{introductionalgorithm}. Using multi-level bucket\cite{goldberg2001a}, it can solve SSSP on graph with edge lengths satisfy nature distribution in expecting linear time. But $Dijkstra's Algorithm$ can only handle nonnegative-weight graphs. 
		\par Thorup's algorithm\cite{Thorup:1999:USS:316542.316548} can solve SSSP on undirected graph in linear time complexity. It is a very important algorithm theoretically. However, it is so complex that it runs actually extremely slow in real life.
		\par Sometimes, we need to solve this problem with arbitrary weights, like solving \textbf{Systems of Difference Constraints}. Negative cycles may appear, so the negative cycle detecting problem is also an important problem. \textit{Bellman-Ford Algorithm}\cite{Bellman1958On} is a basic method to solve it. 	The \textit{queue optimized Bellman-Ford Algorithm}cite{Gilsinn1973A}, which had a complexity of $\Theta(MN)$ as we considered, is an effective improvement. Duan\cite{Duan1994A} claimed that the complexity of \textit{queue optimized Bellman-Ford Algorithm} is $O(M)$, and gave the algorithm a name \textit{Shortest Path Faster Algorithm(SPFA)}\footnote{The following notation SPFA is only for convenience.}. Later the claim was proved wrong. \textit{SPFA} does run fast, when the graph is not specially constructed and no negative cycles are included. But it runs as slow as $O(MN)$ on graphs with negative cycles in expect. It also runs as slow as $O(MN)$ in grid graph, which means it cannot fit a lot of status in real life.

		\par \textbf{My contribution}\footnote{According to \cite{Cherkassky1999Negative}, it seemed that my algorithm is an improvement of Tarjan's Algorithm on a technical report in 1981, but I got the original one independently. However, the linear expectation should not fit for Tarjan's original Algorithm.} I found a simple but effective method to improve \textit{SPFA}. The improved algorithm is so-called \textit{Raffica algorithm}. Considering that the single source shortest path forms a tree, I denote it as an \textit{Auxiliary Tree} and use a breadth-first-like search to maintain it. The relaxing operation in priority-queue-based Dijkstra's algorithm only happens once on each vertex. But it may happen a lot of times on every vertex in Raffica algorithm. Fortunately, using \textit{Auxiliary Tree} and my optimization, it can cut down so many trivial relaxing that each vertex is relaxed expecting $\Theta(1)$ times in random graph. The operation that cuts the relaxing down is called `Raffica`. Predicting the average count of x-depth node, I can control the density of \textit{Raffica} operation. When \textit{Auxiliary Tree} fails to maintain, it turns out to be a cycle on the tree. It means there exists a negative cycle, therefore the problem has no solution.
		\par The configuration model\cite{bender1978the}\cite{bollobas1980a} is a random graph model based on half-edge matching. Define a random multi-graph $G^*(N, (d_i)_1^N)$ with a given degree sequence $(d_i)_1^N$. $d_i$ half-edges are associated to each node $i$. All the half-edges are matched uniformly to become an edge. Therefore, a multi-graph allowing self-loops was created. The weights of the edges are given by a weight sequence. The \textit{Raffica algorithm} can solve the arbitrary-weight SSSP, however, I have to assume the weight is mostly non-negative, else the negative cycle will be found in sublinear time. Except this constraint, the distribution of the weights can be arbitrary.
		\par Negative cycle detecting is also a common model. System of Difference Constraints is an example. Comparing with \textit{SPFA}, \textit{Raffica algorithm} can find the negative cycle once traversed.
		\par \textit{Raffica algorithm} is $O(MN)$ in worst case scenario, while \textit{SPFA} is $O(MN)$ too. However, the worst case scenario of \textit{Raffica algorithm} will hardly appear in traffic problems(in other words, a random near-grid graph), while \textit{SPFA} will easily fall into the worst case scenario by grid graph. The method to improve the worst case of my algorithm is only to reconstruct the graph, or to change the method of searching instead of a breadth first search. However, the common reconstruct method is still an open problem. Though, my algorithm is still practical in traffic problems and random graphs.
		\par Following these conclusions, we can solve the All-Pairs Shortest Path(APSP) problem in $O(MN)$ on random graphs in expect, which is better than \textit{Floyd Algorithm}\cite{floyd1962algorithm}. We can solve the minimum average weight cycle problem by simply using dichotomy and \textit{Raffica algorithm}, therefore we get a $O(M\lg W)$ solution in expect, where $W$ stands the maximum weight of the graph.

		\section{Preliminaries}
		The following graph is a directed weighted graph $G = \left(V, E, W\right)$, $W\in \mathbb{N}$. As an SSSP problem, we denote the source vertex as $S$. $M$ is the count of edges while $N$ is the count of vertices.
		\par\textit{Auxiliary Tree} $T = \left(V, F\right)$ is a tree used in \textit{Raffica algorithm}.
		\par The hop-diameter of the graph is defined as the maximum count of vertices on the shortest path $u\to v$ on each $u, v\in V$. It is denoted as $D$. On a sparse random graph, $D = \log N$ in expect\cite{CornellThe}.
		\par\textit{SP Tree} $T' = \left(V, F\right)$ is the Shortest Path Tree of the SSSP. \textit{Auxiliary} Tree is convergent during the \textit{Raffica algorithm}. Finally, it will be the same as \textit{SP Tree}.
		\par $F_i$ stands the father node of $i$ on any tree.
		\par The depth of a vertex indexed $x$ is denoted as $d_x$ or $depth_x$. It stands the count of vertices on the path from source $S$ to $x$. $d_S = 1$.
		\par Both \textit{Auxiliary Tree} and \textit{SP Tree} are rooted by $S$;
		\par$Inqueue$ is an array saving the label whether the vertex should be in queue, $Inqueue_i$ denotes whether the vertex $i$ is in queue.
		\par$dis_x$ is the distance $S\to x$ of the SSSP problem.
		\par Relaxing is an operation on an edge from \textit{Bellman-Ford Algorithm} and \textit{Dijkstra's Algorithm}. When an edge $E = (u\to v, w)$ is relax-able, it means $dis_u + w < dis_v$.
		\par
		Denote the iteration count of a vertex $X$ is how many vertices it visited from the source vertex $S$ to $x$, in other words, the depth on the \textit{Auxiliary Tree}. An iteration of a BFS-like algorithm(i.e. the following \textit{SPFA} and \textit{Raffica algorithm}) is a series of relaxing where the iteration counts are equal.
		\par During the iteration, I call the vertex in queue as the \textit{Dark Point}.
		\newline\newline\newline\newline\newline\newline
		\begin{figure}[h]
			\centering  
			\includegraphics[width=0.6\linewidth]{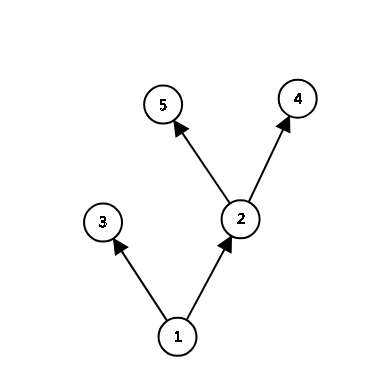}  
			\caption{Example Graph}  
		\end{figure}
		\par For example, if the vertex 1 is the source $S$, then the iteration count of vertex 1 is 1, the count of vertex 2 and 3 is 2, the count of vertex 4 and 5 is 3.
		\par
		The random graph uses the configuration model described in \cite{bollobas1980a} and \cite{bender1978the}. For the node indexed $i\in [1, N]$, the degree sequence $d_i$ describe the count of the half-edges to be associated. Two half-edges will be paired uniformly to form an undirected edge. Though, it is assumed the performance on undirected random graph is similar to directed one. The degree sequence is given by $(d_i)_1^N = {\frac{M}{N}}$. For $W\in \mathbb{N}$, without loss of generality, the weights satisfy i.i.d.(identically distributed) distribution in $(-1, 1)$. We assume the expectation of the weight is higher than 0, and the non-negative weights should occur with high probability(w.h.p.), because the often appearing negative weights will make it exist some negative cycles in averagely a constant neighborhood from $S$, making \textit{Raffica Algorithm} solve in sublinear time complexity.

		\subsection{Bellman-Ford Algorithm}
		\textit{Bellman-Ford Algorithm}\cite{Bellman1958On} is a classic algorithm solving the SSSP on arbitrary weighted graph. It uses relaxing to iterate and after $N-1$ times iteration, we get the answer. If in the $N$th iteration there is any vertex relaxed, those vertices form negative cycle(s).
		\begin{algorithmic}[h]
			\REQUIRE{Edges $E = {u\to v, weight}$, Source $S$, Vertex count $N$, Edge count $M$}
			\ENSURE{Distance $dis$}
			\FOR{$i=1\to N$} 
				\STATE{$dis_i \leftarrow \infty$}
			\ENDFOR
			\FOR{$i=1\to N-1$}
				\FOR{$i=1\to M$}
					\IF{$dis_{Edge_i.u} + Edge_i.w < dis_{Edge_i.v}$}
					\STATE{$dis_{Edge_i.v} \leftarrow dis_{Edge_i.u} + Edge_i.w$}
					\ENDIF
				\ENDFOR
			\ENDFOR
		\end{algorithmic}
		The time complexity is obviously $O(MN)$. There are many kinds of improving methods like Yen's, and so on. The following \textit{SPFA} is also a kind of improving.

		\subsection{SPFA}
		\par\textit{SPFA} is the queue optimized \textit{Bellman-Ford Algorithm}.
		\par\textit{SPFA} uses a queue to keep the vertices, a little bit like BFS. \textit{SPFA} uses Adjacency table. During the \textit{SPFA}, we search and relax, pushing the relax-able vertices into the queue and update the distance.
		\par The following pseudocode describes how \textit{SPFA} works.
		\\\begin{algorithmic}[h]
			\REQUIRE{Edges $E = {u\to v, weight}$, Source $S$, Vertex count $N$, Edge count $M$}
			\ENSURE{Distance $dis$}  
			\STATE{$queue\leftarrow S$}
			\FOR{$i=1\to N$}
			\STATE{$dis_i \leftarrow \infty$}
			\STATE{$inqueue_i\leftarrow false$}
			\STATE{$dis_S \leftarrow 0, inqueue_i \leftarrow true$}
			\ENDFOR
			\WHILE{queue is not empty}
				\STATE{$x \leftarrow queue.front$}
				\STATE{queue.pop}
				\STATE{$inqueue_x \leftarrow false$}
				\FOR{each E adjacent by x}
					\IF{$dis_{E.u} + E.weight < dis_{E.v}$}
						\STATE{$dis_{E.v} \leftarrow dis_{E.u} + E.weight$}
						\IF{($inqueue_{E.v} = false$)}
							\STATE{$inqueue_{E.v} \leftarrow true$}
							\STATE{$queue \leftarrow E.v$}
						\ENDIF
					\ENDIF
				\ENDFOR
			\ENDWHILE
		\end{algorithmic}
		
		\section{Raffica algorithm}
		\textit{Raffica algorithm} is an improved \textit{SPFA}.
		\textit{Raffica algorithm} is based on this theorem:
		\newtheorem* {mypro3}{Theorem}
		\begin{mypro3} {
				The solution of the SSSP forms a tree.
		}\end{mypro3}
		
		\begin{proof}
			
			Obviously the solution includes all the vertices reachable. Suppose the solution of the SSSP includes a cycle, and $X$ is a vertex on the cycle. It means if we go through the cycle from $X$ to $X$, the distance does not increase. If the sum of the weight of the whole cycle is positive, it will not be the SSSP, because going through this cycle will get a worse answer. If the weight of the cycle is zero, we need not go through this cycle. If negative, we will go through the cycle for infinite times, so that the answer does not exist. So the solution of SSSP does not include any cycles.
		\end{proof}
		\par The method is pretty simple. Maintain the \textit{Auxiliary Tree}. When relaxing an edge $E:u\to v$, we set $v$ as $u$'s son on the \textit{Auxiliary Tree} $T$. If $v$ already has a father, we break the edge $v.father\to v$ on the \textit{Auxiliary Tree} and reset $v.father$ as $u$. This operation is called `$Raffica\left(u, v\right)$`. Then we successfully maintain a tree, except that $v$ is an ancestor of $u$ on the tree.
		\par If $v$ is an ancestor of $u$ on the tree, this SSSP has a negative cycle $u\to ... \to u$. When we can relax $E:u\to v$, it means $dis_u + E.w < dis_v$. If this inequality comes to a cycle during the iteration, it only means that there is a negative cycle. Therefore, the SSSP problem has no solutions. So what we maintain is absolutely a tree.
		\par For the BFS can traverse all these vertices reachable, if there is a negative cycle reachable, we can absolutely find it.
		\par So \textbf{finding a cycle on the Auxiliary Tree is the necessary and sufficient condition of existing a reachable negative cycle.}

		\newtheorem*{mypro2}{Theorem}
		\begin{mypro2} Shortest Path(SP) can be divided into smaller SP. \end{mypro2}
		i.e. If SP $u\to v$ includes vertex $a$, we can conclude that $u\to a$ and $a\to v$ on the SP $u\to v$ are also SP in smaller problems.
		\\
		\\During the iteration, suppose there is a vertex $u$ in the queue, and its ancestor $v$ is just \textit{Raffica-ed}. Before\textit{ Raffica-ed}, $dis_v + \left(dis_u - dis_v\right) = dis_u$. $\left(dis_u - dis_v\right)$ means the length of the shortest path from $v$ to $u$. After\textit{ Raffica-ed}, $dis_v$ become lower, so $dis_u$ should become lower too, due to the Theorem 2nd.
		\\For the vertex $u$ is still in queue, it uses an earlier distance data, so we can remove $u$ from the queue, else there will be redundant relaxing. So we clear the in-queue label of $v$'s subtree when $Raffica\left(u, v\right)$.
		\\
		\begin{figure}[h]
			\centering  
			\includegraphics[width=1\linewidth]{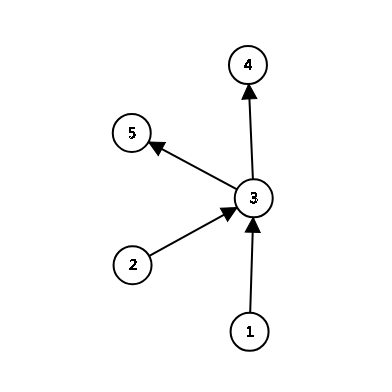}  
			\caption{Example Graph 2}  
		\end{figure}
		\\
		\par Consider the status of the picture. In a relaxation $1\to 3$, supposing 2 is the earlier father of 3. After $Raffica\left(1, 3\right)$, 2 no longer has a child 3. Vertex 1 obtains a child 3.
		\par Suppose that 4 and 5 are in queue at this time, we consider the updating $dis_4$ and $dis_5$. It is still using the old data. In other words, it still believe the best path is $2\to 3\to 4$ and $2\to 3\to 5$. But we have already known $1\to 3$ is better than $2\to 3$, so we should let 4 and 5 out of the queue, while 3 is in the queue, using the newest result, considering $1\to 3$ is the best. In the next iteration, 4 and 5 will consider $1\to 3\to 4$ and $1\to 3\to 5$ as the shortest path.
		\par If we did not do this \textit{Raffica}(i.e. \textit{SPFA}), we will update 4's son and 5's son using $2\to 3\to 4$ and $2\to 3\to 5$ in the first iteration, and using $1\to 3\to 4$ and $1\to 3 \to 5$ in the second iteration. If we would \textit{Raffica} a lot in \textit{Raffica algorithm}, and the \textit{SP Tree} is tall, $SPFA$ will be very slow, because the new data will override the old data during every iteration. Although there is a lot of Rafficas, \textit{Raffica algorithm} would not update all the subtree of \textit{SP Tree} in average. It will update the subtree of \textit{Auxiliary Tree}, which is not very tall.
		\par In fact, it is common that a graph is done many Rafficas while having a tall \textit{SP Tree}: a grid graph is one of the example, where \textit{SPFA} runs slow and \textit{Raffica algorithm} runs fast.
		\par If we need to find out a negative cycle, we consider if \textit{v} is the ancestor of \textit{v}. If yes, there exists a negative cycle. We would simply use a DFS to check it, because DFS costs the same time as Raffica.
		\par For random graphs, I use a different approach to make sure it is linear. It will be shown on Section 5.2.
		\\
		\begin{algorithmic}  
			\REQUIRE{Edges $E = {u\to v, weight}$, Source $S$, Vertex count $N$, Edge count $M$}
			\ENSURE{Distance $dis$}
			\STATE{$queue\leftarrow S$}
			\STATE{$T\leftarrow S$}
			\FOR{$i=1\to N$}
				\STATE{$dis_i \leftarrow \infty$}
				\STATE{$inqueue_i\leftarrow false$}
				\STATE{$dis_S \leftarrow 0, inqueue_S \leftarrow true$}
			\ENDFOR
			\WHILE{queue is not empty}
				\STATE{$x \leftarrow queue.front()$}
				\STATE{queue.pop}
				\IF{$inqueue_x = false $}
					\STATE{continue}
				\ENDIF
				\STATE{$inqueue_x \leftarrow false$}
				\FOR{each E adjacent by x}
					\IF{$dis_{E.u} + E.weight < dis_{E.v}$}
						\STATE{$dis_{E.v} \leftarrow dis_{E.u} + E.weight$}
						\STATE{use a DFS to check if E.v is the ancestor of x}
						\IF{E.v is the ancestor of x}
							\RETURN{Exist a negative cycle}
						\ENDIF
						\STATE{Set it as false that the inqueue label on the vertices on the subtree rooted by E.v}
						\STATE{Clear the father-son relationship of all the vertices on the subtree rooted by E.v}
						\STATE{$T_v.father$ remove child $v$}
						\STATE{$T_v.fa \leftarrow x$}
						\IF{$inqueue_{E.v} = false$}
							\STATE{$inqueue_{E.v} \leftarrow true$}
							\STATE{$queue \leftarrow E.v$}
						\ENDIF
					\ENDIF
				\ENDFOR
			\ENDWHILE
		\end{algorithmic}
		
		\section {Correctness}
		\textbf{Raffica algorithm} 
		\newtheorem*{mypro4}{Theorem}
		\begin{mypro4} For \textit{Raffica algorithm}, the \textit{Auxiliary Tree} is always a \textit{SP Tree} of `the graph consist of all the vertices in \textit{Auxiliary Tree} and the traversed edges`. \end{mypro4}
		\begin{proof}
			I use a Mathematical Induction to prove it will return a correct answer.
			Firstly, a tree consisted of a vertex obviously meets the condition.
			Considering a relax, if it causes no Raffica, obviously it meets the condition. If it causes a $Raffica(u, v)$, the subtree of v is cut, it also meets the condition.
		\end{proof}
		And now I am going to prove this algorithm will come to an end. Fist I prove a vertex cannot be Rafficaed more than \textit{N-2} times.
		Consider a vertex $X$. Except the first time visited, because \textit{Raffica algorithm} goes through the \textit{SP Tree}, BFS costs $N-1$ iterations, which is the maximum possible height of \textit{SP Tree}. For every iteration, only a vertex can Raffica $X$, because in a path from $S$ to a leaf, there is no more than one vertex in the queue.
		
		\par After the Rafficas for every vertices, it remains a BFS. So absolutely the algorithm will come to an end.

		\section {Time Complexity}
		\subsection{Worst Case}
		The worst case scenario of the \textit{Raffica algorithm} is shown on the following picture:
		\\
		\begin{figure}[h]
			\centering  
			\includegraphics[width=.7\linewidth]{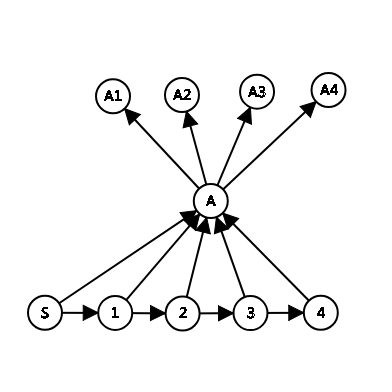}  
			\caption{Worst case}  
		\end{figure}
		\par Point A has $O(M)$ output and will be \textit{Raffica-ed} $O(N)$ times. We may reconstruct the graph by spliting the output or randomizing. The worst case scenario time complexity can be improved to $O(\dfrac{MN}{\log N})$, but it is trivial.
		\par In real life, there may be little vertices with $O(N)$ output like the above figure. In Section 8, I will show that this case can be improved. The unimprovable worst case scenario appears in a desperately extreme status, unlike the worst case of SPFA.
		\\
		\\\subsection{Random Graph}
		
		\par \textbf{Raffica algorithm}
		A BFS is obviously $O(M)$, we consider the extra complexity caused by \textit{Raffica}.
		\par Firstly, the count of \textit{Raffica} is no more than $M$(due to the correctness proof).
		
		\par On the \textit{Auxiliary Tree}, one's ancestor can not be in queue with it. Suppose there is a leaf vertex $X$. $S\to X$ is a series of edges from root to leaf.
		\newtheorem*{mypro}{Theorem}
		\begin{mypro} When an edge $E=\left(u, v\right)$ in the final SP tree is accessed, $u$ and $v$ will not be Raffica-ed. \end{mypro}
		\begin{proof}
			There will not be any path shorter than the final \textit{SP Tree}. Once it is accessed, there are no solutions better than this solution. So both $u$ and $v$ will not be Raffica-ed.
		\end{proof}
		
		Due to \cite{1609.07269}, there is a property on the neighborhoods from $S$:\\
		\begin{equation*}
		\begin{aligned}
		W_n(u, v) &- \frac{2\log\log N}{|\log(\tau-2)|}\\
		H_n(u, v) &- \frac{2\log\log N}{|\log(\tau-2)|}\\
		D_n(u, v) &- \frac{2\log\log N}{|\log(\tau-2)|}, \tau\in(2, 3)
		\end{aligned}
		\end{equation*}
		are all tight sequence of random variables, which means they are convergent to 0 w.h.p.\\
		Here $W_n(u, v)$ stands for the distance between $u$ and $v$ on a graph consist of $n$ nodes, where $u$ and $v$ are uniformly chosen in [1, n]. $H_n(u, v)$ stands for the hop-distance between $u$ and $v$(hop-distance is defined by the count of edges on the shortest path). $D_n(u, v)$ stands for the distance between $u$ and $v$ without consideration of weight. It shows that the vertices are mainly distributed in a small interval on the hop-distance and distance.\\
		Due to \cite{Amini2012The}, the diameter is $O(\log N)$.
		\\
		
		Consider each path from the root to the leaf.
		Basically, the density of \textit{Raffica} is defined by the count of iterations between 2 Rafficas. The extra cost of $Raffica(u, v)$ is the size of the subtree rooted by $v$. The total cost of the \textit{Raffica} operation is the sum of all the extra costs.\\
		
		\begin{figure}[h]
			\centering
			\includegraphics[width=0.7\linewidth]{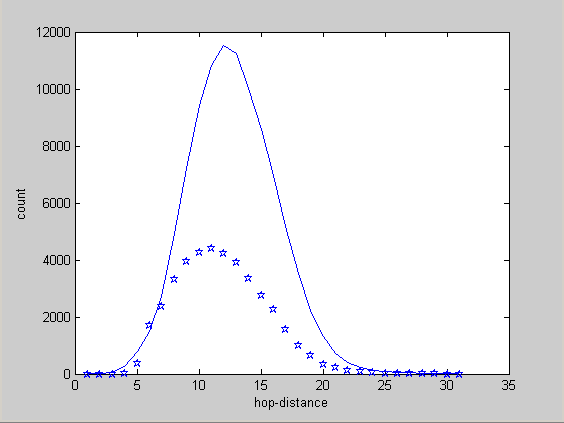}
			\caption{N = $10^5$, M = $10^6$}
		\end{figure}
		\begin{figure}[h]
			\centering
			\includegraphics[width=0.7\linewidth]{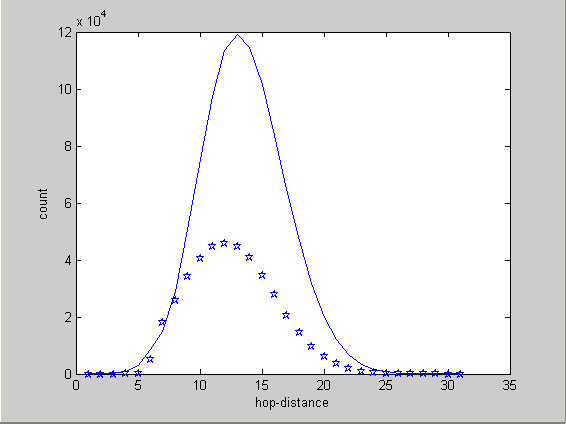}
			\caption{N = $10^6$, M = $10^7$}
		\end{figure}
		The figure 4 and 5 are my test cases. The solid curve is the count of x-depth vertices, while the dotted curve is the count of \textit{Raffica} on x-depth.\\
		
		Basically, the density of \textit{Raffica} on a sparse graph is $O(\dfrac{1}{\log N})$ w.h.p.\footnote{This result will not be used in my proof, so I make a claim without proof. For a dense graph, it can always be transformed into a sparse random graph according to Section 8.} It is not very good for my algorithm. Now I introduce my improvement mentioned above.
		\newtheorem*{claimdensity}{Theorem R}
		\begin{claimdensity}
			Suppose the density of \textit{Raffica} is a constant w.h.p., our algorithm should be linear in expectation.
		\end{claimdensity} 
	Therefore, the improvement is:\par
		For it is a random graph, we can predict the average distribution of x-depth vertices and x-depth \textit{Raffica} counts before the algorithm. We randomly add some \textit{Raffica} due to the distribution, so that the density tends to 0.5 in probability. More precisely, we select some edges that cannot fit the triangle inequality during the iteration, and disassemble its subtree. The two vertices connected by the edge should be on the adjacent depth. Those chosen edges would be \textit{Raffica}ed due to the selection of probability.
		\par It seems that adding more \textit{Raffica} will simply make the algorithm slower. However, it cuts down the average size of the subtree, making its total average time shorter.
		\\\textbf{Proof of Theorem R}
		
		Denote the count of vertices that is on the i-depth from S as $\Gamma_{i}(S)$, and $\Gamma'_{i}(S) = \dfrac{\Gamma_{i}(S)}{\Gamma_{i-1}(S)}$. Easily, $\Gamma'_{i}(S)$ is monotone decreasing with variable $i$.\\
		\newtheorem*{lemmatonari}{Lemma 1}
		\begin{lemmatonari}
			$\mathbb{P}(\Gamma'_{dis_u}(S) < \epsilon_0) = 1$, for any $\epsilon_0>1$ and $u$ uniformly chosen in $[1, N]$.
		\end{lemmatonari}
		\begin{proof}
			Due to \cite{Adriaans2017Weighted}, $dis_u$ tends to infinity w.h.p.
			$T_a(i)$ was defined by the radius of the ball centered by $S$ that contains $k$ vertices. More precisely,
			\begin{equation*}
			T_a(i) = r|\{x\in [1, N]|dis_x <= r\}| = i
			\end{equation*}
			 According to \cite{CornellThe}, for $\alpha_N = \lfloor \log^3 N \rfloor$, $\beta_N = \lfloor \sqrt{\dfrac{\mu}{v-1}N\log N} \rfloor$(two balls of size at least $\beta_N$ intersect almost surely), $\alpha_N \leqslant i \leqslant \beta_N$, some constant number $v, \epsilon, \mu$, define $\rho = \dfrac{v-1}{1+\epsilon}$,
			\begin{equation*}
			\begin{aligned}
			T_a(i+1) - T_a(i) &\leqslant e^{\frac{v-1}{1+\epsilon}i}\\
			&= i^{3\rho}
			\end{aligned}
			\end{equation*}
			For $\Gamma'_{i}(S)$ is monotone decreasing with $i$, $T_a(i)$ is monotone increasing with $i$. We would only consider $T_a(\alpha_N)$, for the corresponding $\Gamma'_{depth_i}$ is largest. Its inverse function is:
			\begin{equation*}
			T_a^{-1}(i) = i^{\frac{1}{3\rho}}
			\end{equation*}
			The inverse function satisfies $T_a^{-1}(i) \propto \Gamma_{depth_i}$. Therefore,
			\begin{equation*}
			\begin{aligned}
			i &= \alpha_N\\
			\Gamma'_{depth_i} &= \frac{\Gamma_{depth_i}}{\Gamma_{depth_i-1}}\\
			&=\frac{T_a^{-1}(i)}{T_a^{-1}(i-1)}\\
			&=\frac{i^{\frac{1}{3\rho}}}{(i-1)^{\frac{1}{3\rho}}}
			&=1
			\end{aligned}
			\end{equation*}
			Meanwhile, $\mathbb{P}(dis_u < T_a(\alpha_N)) = 0$; when $dis_u > T_a(\beta_N)$, $\Gamma'_{depth_i} < 1$. Q.E.D.

		\end{proof}
		
		The expecting time of \textit{Raffica} is:
		\begin{equation*}
		\begin{aligned}
		T_{Raffica}(N) &= O(M/N *  \sum_{i=1}^{N}\frac{\Gamma_{depth_i-2}(S)}{\Gamma_{depth_i}(S)})\\
		&\leqslant O(M/N * N * \max{\frac{\Gamma_{depth_i-2}(S)}{\Gamma_{depth_i}(S)}})\\
		&= O(M/N) * O(N)\\
		&= O(M)
		\end{aligned}
		\end{equation*}
		\\
		\par It is worth mentioning that this optimization should not work better than the original one on small test cases. This optimization always makes the enter-queue count tend to $O(M)$, while the original one makes the enter-queue count tend to $O(M\log N)$. The factor under the signal $O$ causes this phenomenon.

		\par Now we know that the total cost by \textit{Raffica} is $O(M)$. The question turns out to be how much time does other operations cost. It looks like very slow to check and update a subtree of $E.v$, because each DFS may cost $O(M)$. However, the subtree to be maintained and to be checked, is the same size as \textit{Raffica} operator, making the other operations using DFS the same cost as \textit{Raffica}. \textbf{The conclusion is, \textit{Raffica Algorithm} has $\Theta(M)$ time complexity in expect on random graph.}
		\newline
		\\
		\textbf{SPFA} For \textit{SPFA}, I give an upper bound. The difference between $SPFA$ and \textit{Raffica algorithm} is \textit{SPFA} does not clear the in-queue label of the subtree of the \textit{Dark Point}. According to \cite{Amini2012The}, the diameter of a random graph is $D = O(\log N)$. The upper bound is $O(M\log N)$.
		
		\par When \textit{SPFA} deals with negative cycle, its expect time complexity is $\Theta(MN)$, while \textit{Raffica algorithm} is $\Theta(M)$. Because \textit{SPFA} judges a negative cycle by checking how many times any vertices be in queue. If one enters the queue $N$ times, there turns out to be a negative cycle. For each iteration, there are totally expecting $\Theta(N)$ vertices in queue. \textit{Raffica algorithm} gets the existence of a negative cycle when it traverses a negative cycle. Therefore, the time complexity is $\Theta(M)$.
		\\
		\\
		\subsection{Traffic Problem or Grid Graph}
		
		\begin{figure}[h]
			\centering  
			\includegraphics[width=0.7\linewidth]{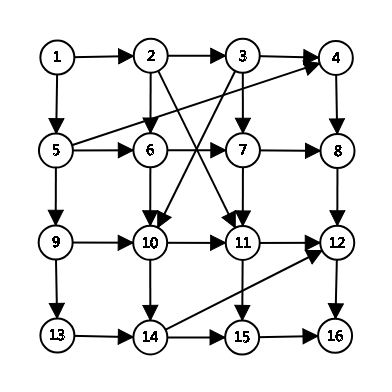}  
			\caption{A grid graph}  
		\end{figure}
		
		\par Now we consider a grid or a near-grid graph. The diameter of these graphs are $\Theta(N)$, and the counts of out degrees are $O(1)$. The weights of the edges satisfy the same distribution as the model of above random graph. A grid graph or near-grid graph is often seen in real life. I call it a traffic problem.
		\par\textit{SPFA} runs slow in this graph, while \textit{Raffica algorithm} runs in linear complexity.
		\par Using a similar way, when $N$ tends to infinity, the density of \textit{Raffica} is made to be $0.5$. In the same way, the time is $\Theta(M)$.
		\par Consider \textit{SPFA}. The diameter of the grid graph is $\Theta(N)$, therefore the time of \textit{Raffica} is $\Theta(N)$. The density of \textit{Raffica} tends to $\widetilde O(1)$. The total cost is $\Theta(N) * \Theta(N) * \widetilde O(1) = \widetilde\Theta(N^2)$ in expect.
		\par In fact, \textit{SPFA}'s time complexity depends on the height of the \textit{SP tree} and the density of \textit{Raffica}. Actually, the grid graph is not the only one of the graphs that \textit{SPFA} runs slow. If both the height and the density are huge, \textit{SPFA} also runs slow. This status often appears in traffic SSSP problems.
		
		\section{Comparison}
		\begin{table*}[h]
			\centering
			\caption{SSSP}
			\begin{tabular*}{500pt}{@{\extracolsep{\fill}}|c|c|c|c|c|c|}
				\hline
				&Dijkstra(Fib Heap)	&Bellman-Ford	&SPFA			&Raffica	&Thorup\cite{Thorup:1999:USS:316542.316548}\\
				\hline
				Non-negative weighted Random Graph	&$O(M+N\lg N)$		&$O(MN)$		&$O(MD)$	&$O(M)$ &$O(M)$\\
				\hline
				Worst Case		&$O(M+N\lg N)$		&$O(MN)$		&$O(MN)$		&$O(MN)$ &$O(M)$\\
				\hline
				Negative Cycle
				Random Graph	&unable				&$O(MN)$		&$O(MN)$		&$O(M)$ &unable\\
				\hline
				Arbitrary Weight
				Random Graph	&unable				&$O(MN)$		&$O(MD)$	&$O(M)$ &unable\\
				\hline
				Traffic Problem\footnotemark	&$O(M+N\lg N)$		&$O(MN)$		&$O(MN)$		&$O(M)$ &$O(M)$\\
				\hline
				Directed Graph	&able				&able			&able			&able	&unable\\
				\hline
			\end{tabular*}
		\end{table*}
		\begin{table*}[h]
			\caption{APSP}
			\centering
			\begin{tabular*}{500pt}{@{\extracolsep{\fill}}|c|c|c|c|c|c|}
				
				\hline
				&Hagerup\cite{hagerup2000improved}	&Floyd		&SPFA			&Raffica &Thorup\\
				\hline
				Non-negative Weighted Random Graph	&$O(MN+N^2\lg\lg N)$		&$O(N^3)$		&$O(MND)$	&$O(MN)$ &$O(MN)$\\
				\hline
				Worst Case		&$O(MN+N^2\lg\lg N)$		&$O(N^3)$		&$O(N^2M)$	&$O(N^2M)$ &$O(MN)$\\
				\hline
				Arbitrary Weight
				Random Graph	&$O(MN+N^2\lg\lg N)$		&$O(N^3)$	&$O(MND)$	&$O(MN)$ &unable\\
				\hline
				Directed Graph	&able				&able			&able			&able	&unable\\
				\hline
			\end{tabular*}
		\end{table*}
		I use two tables to compare the performance of some classical algorithms and my algorithm on SSSP and APSP.
		\par In random graph or traffic problem, Raffica algorithm has linear complexity, which is absolutely fastest. Thorup's algorithm is also linear, which should be the fastest too. But it can only handle undirected graph, and it is very complex. 
		
		\section {Applications}
		\subsection{System of Difference Constraints}
		\par System of Difference Constraints is a problem handling a series of inequality $X_i - X_j < k$. It can be easily transformed to a SSSP problem. It is widely used to many applications, such as temporal reasoning.
		\par Set a vertex S as the super source vertex. Transform the inequality $X_i - X_j < k$ to $X_i + \left(-k\right) < X_j$. For each inequality, add an edge $X_i\to X_j$ with weight $-k$. For each vertex X, add an edge $S\to X$ with weight 0. Then regard S as source, the solution of System of Difference Constraints is the solution of SSSP. If there exists any negative cycle, the solution does not exist.
		\begin{figure}[h]
			\centering  
			\includegraphics[width=.7\linewidth]{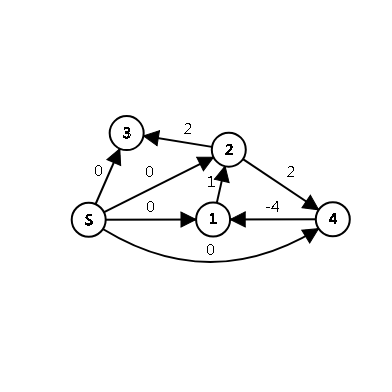}  
			\caption{System of Difference Constraints}  
		\end{figure}
		
		The above graph stands for these constraints:
		\begin{equation*}
		\begin{aligned}
		x_1-x_2&<-1\\
		x_2-x_3&<-2\\
		x_2-x_4&<-2\\
		x_4-x_1&<4
		\end{aligned}
		\end{equation*}
		\\
		\par This system has no solutions because $1\to 2\to 4\to 1$ is a negative cycle.
		\\\par
		We want to know if this problem has any solutions. Therefore we want to find if this SSSP problem has a negative cycle.\par
		\textit{Bellman-Ford Algorithm} and \textit{SPFA} cannot solve the find-negative-cycle problem very fast. If the problem is random one,  \textit{Raffica algorithm} can solve it in linear time.

		\subsection{Detecting the minimum average weight cycle}
		An average weight of a cycle denotes the total sum of the weights of the cycle divided by the total count of the cycle. Karp's algorithm\cite{KARP1978309} solves it in $O(MN)$. If the graph can be regarded as a random graph, \textit{Raffica algorithm} can solve it in $O(M\lg W)$($W$ stands the max weight of the edges) by a simple dichotomy.

		\section{Open problem}
		\begin{figure}[h]
			\centering  
			\includegraphics[width=.6\linewidth]{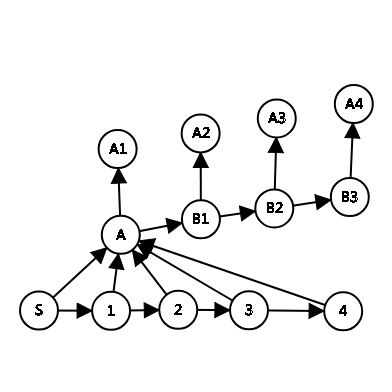}  
			\caption{Reconstructed Worst case}  
		\end{figure}
		\begin{figure}[h]
			\centering  
			\includegraphics[width=.6\linewidth]{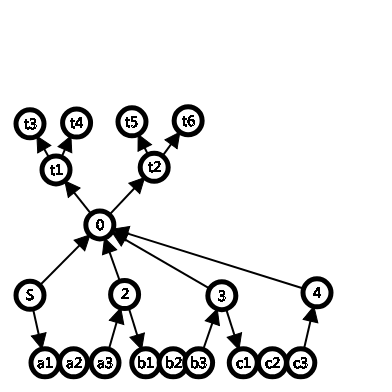}  
			\caption{Worst case that can hardly be reconstructed}  
		\end{figure}
		\par Figure 8 is the worst case scenario. This case may often appear in real life. We can easily transform it to the figure 6 scenario. For the vertex with many out degrees, we separate these edges into those vertices. It is easy to see the correctness of the transformation. And it is easier(linear) to solve by \textit{Raffica algorithm}.
		\\
		\par Figure 9 is another status. The 0 vertex has a subtree looked like a binary-tree. We cannot handle it like the upper one. But we can hardly see it in real life. What we can do is to change the method of searching: not BFS or DFS, but an IDFS. I cannot quantitative the effect of this optimization yet.
		\par In conclusion, if there is a vertex with a subtree on the \textit{SP Tree} having $O(N)$ vertices in a particular depth, and we visit it in a particular way, so that those $O(N)$ vertices updated many times. And that is the worst case.
		\par Another feature of the worst case is: there are some vertices visited many times. But even if we separate the in-degree, the in-degree may also appear like a binary-tree.
		\par In this way, we can separate the in-degree and out-degree, improving the worst case to $\Theta(MN/\log N)$. Reconstructing the graph remains an open problem.
		\par We can also use priority queue to improve the worst case. Using an evaluation function, we can let it have a higher priority where the size of sub-tree is small. It can solve the binary-tree status, but it can not tackle all the statuses.

		\section{Acknowledgments}
		I would like to acknowledge Mr.Wang, Mr.Wen and anyone who has supported me, and those who used this algorithm to real problem settling, giving me priceless test cases and examples. Thanks to Prof.Andrew for informing me my algorithm is very similar to Tarjan's unpublished algorithm, to Prof.Tarjan for pointing out my theoretical mistakes. Thanks to Mr.Zhou for checking language mistakes.
		
		\bibliographystyle{abbrv}
		\bibliography{bib}  
		
		\nocite{lotstedt1986numerical}
		\nocite{Edmonds1972Theoretical}
		\nocite{Ravindra1990Faster}
		\nocite{Cherkassky1999Negative}
		\nocite{Goldberg1993Scaling}
		\nocite{Andrew1993A}
		\nocite{Boykov2002Fast}
		\nocite{Cherkassky2010Shortest}
		\nocite{Pantziou1992Efficient}
		\nocite{Krishnendu2014Approximating}
		\nocite{goldberg1995scaling}
		\nocite{Andrew1993A}
		\nocite{Ravindra1990Faster}
		\nocite{Lysgaard1995A}
		\nocite{unknown}

	\end{document}